\documentclass[preprint, authoryear, times]{elsarticle}
\journal{arXiv}
\usepackage[T1]{fontenc} %
\usepackage[utf8]{inputenc} %
\usepackage{color} %
\usepackage{url}
\usepackage{amsthm}
\usepackage{amsmath} %
\usepackage{amssymb} %
\usepackage{nameref}
\usepackage{graphics} %
\usepackage{graphicx} %
\usepackage{bm}
\usepackage{subfig} %
\usepackage{booktabs}
\usepackage{lmodern} %
\usepackage{setspace}
\usepackage{array}

\providecommand{\assumptionname}{Assumption}
\providecommand{\lemmaname}{Lemma}
\providecommand{\remarkname}{Remark}
\providecommand{\theoremname}{Theorem}
\providecommand{\definitionname}{Definition}
\providecommand{\notationname}{Notation}

\theoremstyle{plain} \newtheorem{thm}{\protect\theoremname}
\theoremstyle{remark} \newtheorem{rem}{\protect\remarkname}
\theoremstyle{plain} \newtheorem{lem}{\protect\lemmaname}
\theoremstyle{plain} \newtheorem{assumption}{\protect\assumptionname}
\theoremstyle{definition} \newtheorem{defn}{\protect\definitionname}
\theoremstyle{definition} \newtheorem{notation}{\protect\notationname}

\newcommand{\E}{\mathbb{E}} %
\renewcommand{\P}{\mathbb{P}} %
\newcommand{\R}{\mathbb{R}} %
\newcommand{\V}{\operatorname{Var}} %
\newcommand{\C}{\operatorname{Cov}} %
\newcommand{\B}{\operatorname{Bias}} %
\newcommand{\argmin}{\mathop{\textrm{argmin}}} %

\begin{document}

\begin{frontmatter}
  \date{}

  \title{Rates of convergence in conditional covariance matrix
    with nonparametric entries estimation}

  \author[jm]{Jean-Michel Loubes} %
  \ead{loubes@math.univ-toulouse.fr}

  \author[cm]{Clement Marteau} %
  \ead{marteau@math.univ-lyon1.fr}

  \author[ms]{Maikol Sol\'is \corref{cor1}} %
  \ead{maikol.solis@ucr.ac.cr} \cortext[cor1]{Corresponding author}

  \address[jm]{Institut de Math\'ematiques de Toulouse, Toulouse,
    France.} \address[cm]{Institut Camille Jordan, Lyon, France.}
  \address[ms]{Universidad de Costa Rica, San Jos\'e, Costa Rica. }

  \begin{abstract}
    Let $X\in \R^p$ and $Y\in \R$ be two random variables. We estimate the conditional covariance matrix $\C\left(\E\left[\boldsymbol{X}\vert Y\right]\right)$ applying a plug-in kernel-based algorithm to its entries. Next, we investigate the estimators rate of convergence under smoothness hypotheses on the density function of $(\boldsymbol{X},Y)$. In a high-dimensional context, we improve the consistency the whole matrix estimator by providing an decreasing structure over the $\C\left(\E\left[\boldsymbol{X}\vert Y\right]\right)$ entries. We illustrate a sliced inverse regression setting for time series matching the conditions of our estimator.
  \end{abstract}


\end{frontmatter}


\section{Introduction}
\label{sec:intro}

The multivariate regression models the relationship between two random variables, the output $Y\in\R$ and the inputs $\boldsymbol{X}\in \R^p$. With the advance of modern acquisition devices, the number $p$ is large, which increase the model's complexity. Extracting relevant patterns between $\boldsymbol{X}$ and $Y$ has become a source of theoretical and practical studies in different areas.

The most common example is given by the regression model
\begin{equation}
  \label{eq:sir}
  Y = \psi(\boldsymbol{X}) + \varepsilon,
\end{equation}44
where $\psi$ is an unknown function. We could guess a parametric form for the data and adjust it with least squared errors or maximum likelihood. However, the \textit{curse of dimensionality} occurs when the number of variables exceeds the number of observations. Namely, when $p$ gets larger than the number of data $n$, the volume of empty spaces growth. The modeling capacity in this scenario diminishes, causing unpredictable conclusions. The paper of~\cite{cook2007dimension} reviews some techniques of dimension reduction in regression.

In a seminal contribution,~\cite{li1991sliced} considered the regression model
\begin{equation}
  \label{eq:sir-simplified}
  Y = \varphi\left( \nu_{1}^{\top} \boldsymbol{X}, \ldots, \nu_{k}^{\top}
    \boldsymbol{X}, \varepsilon \right),
\end{equation}
where $k$ is much less than $p$, which is denoted as $k\ll p$. The $\nu$'s are unknown fixed vectors, $\varepsilon$ is independent of $\boldsymbol{X}$, and $\varphi$ is a $\R^{k + 1}$ arbitrary real valued function. By projection, this model implies the extraction of all the relevant information for $Y$ using only a $k$-dimensional subspace generated by the $\nu$'s. These directions are called effective dimension reduction directions.

The main idea of the sliced inverse regression method is to estimate the unknown matrix
\begin{equation}
  \label{eq:def_Cov_E_X|Y}
  \Sigma = \C\left( \E\left[\boldsymbol{X}\vert Y\right]\right) = \left(\sigma_{ij}\right)_{p\times p},
\end{equation}
where we denote $\sigma_{ij}$ as the $(i,j)$ matrix element. This matrix is degenerate in any direction orthogonal to the $\nu$s. Therefore, the eigenvectors, $\nu_j$ ($j=1,\ldots,k$), associated with the largest $k$ eigenvalues of $\Sigma$ are the effective dimension reduction directions. This leads to the classical sliced inverse regression method. It slices the inverse regression curve, approximates its empirical covariance from those slices, and then estimates the largest eigenvalues with their corresponding eigenvectors. The first $k$ eigenvectors span the effective dimension reduction subspace.




The following authors have contributed to the estimation of $\Sigma$. For example, \citet{hsing1999nearest} worked with nearest neighbors and sliced inverse regression. \citet{setodji2004k} and \citet{cook2005sufficient} transformed the sliced inverse regression into a least squares minimization problem using a $k$-means algorithm. \citet{ferre2003functional}, \citet{ferre2005smoothed}, and \citet{zhu1996asymptotics} developed nonparametric methods involving kernel estimators to model $\Sigma$. \citet{bura2001estimating} assumed some parametric form for $\E\left[\boldsymbol{X}\vert Y\right]$. Alternatively, \citet{solis2011efficient} used a functional Taylor approximation on $\C\left(\E\left[\boldsymbol{X}\vert Y\right]\right)$ to get an efficient estimate.

Works like \cite{Becker2002}, \cite{Gather2004} and \cite{Dahlhaus2000} points out investigations in healthcare models. Those models are based in time series and require a particular covariance matrix. In Section~\ref{sec:sir-time-series} we will discuss how to formulate the problem in this context and point out some applications.



Given an i.i.d.\ sample ${(\boldsymbol{X}_i,Y_i)}_{i=1\dots n}$, our aim is to build an estimator of $\Sigma$ when the joint density of $(\boldsymbol{X},Y)$ is unknown.
To achieve this aim, we plug a marginal density estimate into the conditional covariance parametric estimator and study its asymptotic behavior. Under some smoothness assumptions, we consider the marginal density of $Y$ as a nuisance parameter, but without perturbing the convergence rate for the covariance.

We use the nonparametric estimator developed by \citet{zhu1996asymptotics} to compute the elements entrywise for $\Sigma=(\sigma_{ij})$. Thanks to a kernel estimator we can recover the unknown marginal density of $Y$ and the vector of conditional densities $\E[\boldsymbol{X}\vert Y]$. We propose a new estimator for the conditional covariance matrix, $\Sigma$, based on a plug-in version of the banding estimator. We use the normalized Frobenius norm to measure the squared risk over a class of matrices. Provided that the model is regular enough, it is possible to achieve a pointwise parametric rate of convergence for the estimator of $\sigma_{ij}$. The estimator of $\Sigma$ has a parametric behavior with respect to the Frobenius norm. In these cases, the conditional covariance matrix $\Sigma$ estimator becomes an efficient semiparametric issue.

If $p$ is larger than $n$, we face high-dimensional issues. Thus, our estimator $\hat\Sigma$ of the matrix $\Sigma$ will have unexpected features, e.g., a lack of consistency or significant spreading of the eigenvalues. In a different context, we refer to \citet{marcenko1967distribution}, \citet{johnstone2001distribution}, and their references. We will suggest a banding regularization (see \citet{bickel2008regularized}) to avoid any inconsistency with the rate of convergence.



The rest of this paper is organized as follows. Section~\ref{sec:Methodology} describes the nonparametric algorithm introduced by \citet{zhu1996asymptotics} for the estimation of $\sigma_{ij}$. In Section~\ref{sec:Hypothesis}, we present all of the assumptions required for the consistency and convergence of our estimator. The main focus of this study is given in Section \ref{sec:upper-bound-element}, where we present the convergence rate for the estimator $\hat{\sigma_{ij}}$. As an additional contribution, in Section \ref{sec:consistency-under-frob-norm}, we extend our study to the whole matrix by assuming a particular arrangement for $\Sigma$. One application of SIR and time series is presented in Section~\ref{sec:sir-time-series}. Finally, we give our conclusions in Section \ref{sec:Conclusions}. All of the technical proofs are gathered in Section \ref{sec:Appendix}.




\section{Methodology for nonparametric conditional covariance
  estimation}\label{sec:Methodology}

Let $\boldsymbol{X} = (X_1, \ldots, X_p) \in\R^{p}$ be a random vector and $Y\in\R$ be a random variable. We denote $f_{i}(x_{i},y)$ the joint density of the couple $(X_{i},Y)$. Let $f_Y (\cdot) = \int_{\R} f_{i}(x_{i},\cdot)dx_{i}$ be the marginal density function with respect to $Y$.

Suppose that $\boldsymbol{X}_k^{\top} = (X_{1k}, \ldots, X_{pk})$ and $Y_k$ for $k=1,\ldots,n$ are i.i.d.$\!$ observations with the same law as the random vector $\boldsymbol{X}^\top = (X_{1}, \ldots, X_{p})$ and $Y$. Without loss of generality, we suppose that $\E[X_i]=0$ for $i=1,\dots,p$.

Based on the sample $\mathcal{S} =(\boldsymbol{X}_k,Y_k)_{k=1..n}$, our aim is to estimate the covariance matrix
\begin{equation*} 
  \Sigma = \left(\sigma_{ij} \right)_{i, j=1, \dots, p} =
  \left(\C(\E(X_i \vert Y), \E(X_j\vert Y))\right)_{i, j=1, \dots, p}.
\end{equation*}
We will use a nonparametric method in an inverse regression framework, based on the work of \cite{zhu1996asymptotics}. Following this paper, we introduce the notation,
\begin{align}
  R_i(Y) %
  & = \E[X_{i}\vert Y] \nonumber\\
  \begin{split}\label{eq:g_i}
    g_i(Y) & = R_i(Y)f_Y(Y) = \int_\mathbb{R} x_i f_{i}(x_i,Y)dx_i \quad \forall i \in \{1,\ldots,p\}.
  \end{split}
\end{align}

For any $i,j \in \{1,\ldots,p\}$, we write the $(i,j)$-entry of $\Sigma$ as
\begin{align*}
  \sigma_{ij}
  & =\E[\E[X_{i}\vert Y]\E[X_{i}\vert Y]] =\int_\mathbb{R}
    \frac{g_{i}(y)g_{j}(y)}{f^2_Y(y)}f_Y(y)dy.
\end{align*}
which involves two unknown functions: $g_i(\cdot)$ and $f_Y(\cdot)$. We will estimate and control these functions under some assumptions described in Section \ref{sec:Hypothesis}.

For the sake of simplicity, assume in a first time $f_Y(\cdot)$ known. We can preliminary estimate $\sigma_{ij}$ as
\begin{equation}
  \tilde{\sigma}_{ij} %
  \frac{1}{n} \sum_{k=1}^{n} \frac{\tilde{g}_{i}(Y_{k}) \tilde{g}_{j}(Y_{k})}{f_Y^{2}(Y_{k})},
  \label{eq:def_sigma_f_ij}
\end{equation}
where
\begin{equation*}
  \tilde{g}_{i}(Y_{k}) = \frac{1}{nh} \sum_{l=1}^{n} X_{il} K\left(\frac{Y_{k}
      -Y_{l}}{h}\right).
\end{equation*}

Here, $K(u)$ is a kernel function that satisfies some conditions (see Assumption \ref{ass:kernel} below), and $h$ is a bandwidth that depends on $n$.

The drawback with this approach is that we use twice the same observations. Indeed, we need the whole sample to estimate first the quantities $\tilde{g}_i$ and $\tilde{g}_j$, and then to estimate $\tilde{\sigma}_{ij}$. Thus, dependency issues arise for the estimation of these functions. which entails an over-adjustment by using training data as testing data in the model. To correct this, we change the estimator $\tilde{g}_{i}(Y_{k})$ by removing the $k^{\text{th}}$ observation as follows,
\begin{equation*}
  \hat{g}_{i}(Y_{k}) %
  = \frac{1}{(n-1)h} \sum_{\substack{l=1  \\ l
      \neq k}}^{n} X_{il} K\left(\frac{Y_{k}-Y_{l}}{h}\right).
\end{equation*}


Following the same principle, the next step is to replace the function $f_Y(\cdot)$ by its nonparametric estimator in equation \eqref{eq:def_sigma_f_ij}, i.e.,
\begin{equation*}
  \hat{f}_Y(Y_{k})
  = \frac{1}{(n-1)h} \sum_{\substack{l=1 \\ l \neq k}}^{n}
  K\left(\frac{Y - Y_{l}}{h}\right).
\end{equation*}

However, to avoid dividing by small values in the denominator, we will replace $\hat{f}_{Y}(y)$ by a corrected version $\hat{f}_{Y,b}(y) = \max\{\hat{f}_Y(y),b\}$ with $b$ given by some conditions (see Assumption \ref{ass:f-small-values} and Remark~\ref{ass:order_conv_h_b}).


Therefore, we propose the following estimator for $\sigma_{ij}$,
\begin{equation}
  \label{eq:def_estimator_Sigma_b}
  \hat{\sigma}_{ij} = \frac{1}{n} \sum_{k=1}^{n}
  \frac{\hat{g}_{i}(Y_{k})
    \hat{g}_{j}(Y_{k})}{\hat{f}_{Y,b}^{2}(Y_{k})}.
\end{equation}

Equation \eqref{eq:def_estimator_Sigma_b} presents a robust estimator for $\sigma_{ij}$ due to all consideration already exposed. We shall use it for the rest of the article. In some sense, we consider a semiparametric framework in this contribution. Our aim is to study the conditions under which the density of $Y$ is a mere blurring parameter that plays none role in the estimation procedure. Here, the plug-in method does not hinder the estimation rate for the conditional covariance $\Sigma$, thereby leading to an efficient estimation rate. In the next section, we establish the rate of convergence for the mean squared risk of $\hat{\sigma}_{ij}$.

\section{Main theoretical results}

\subsection{Assumptions}%
\label{sec:Hypothesis}

Here and belowe, $C$, $C_{1}$, and $C_{2}$ denote constants (independent of $n$), which may take different values.

In the following, we assume that $f(x,y)$ has compact support. Let $\beta$ a positive parameter controlling the smoothness for marginal probability functions of $f$. We need to impose a formal structure for $f(x,y)$ related with $\beta$, to guarantee parametric consistency ---convergence in the order of $n^{-1}$--- for our model. Definition~\ref{def:Holder} links the parameter $\beta$ with the number of finite derivatives of those marginals, through a H\"older class of smooth functions. Define $\lfloor\beta\rfloor$ as the largest integer such that $\lfloor\beta\rfloor\leq\beta$.

\begin{defn}
  \label{def:Holder}
  Denote $\mathcal{H}(\beta,L)$ as the \textit{H\"older class} of density functions with smoothness $\beta>0$ and radius $L>0$, which is defined as the set of $\left\lfloor \beta\right\rfloor$ times differentiable functions $\phi:T\to\R$, where $T$ is a compact interval in $\R$, for which the derivative $\phi^{(\lfloor\beta\rfloor)}$ satisfies

  \begin{equation*}
    \left| \phi^{\left(\lfloor\beta\rfloor\right)}(x) -
      \phi^{\left(\lfloor\beta\rfloor\right)}(x^{\prime}) \right| \leq L \left|
      x - x^{\prime} \right|^{\beta-\lfloor\beta\rfloor}, \ \forall x,
    x^{\prime} \in T.
  \end{equation*}

\end{defn}

The proofs use the following assumption continuously.

\begin{assumption}\label{ass:f_g_Holder}
  For fixed $x_{i}$ $i=1,\ldots,p$, the function $f_{i}(x_{i},y)$ belongs to a H\"older class of regularity $\beta$ and constant $L$ uniformly for each $i$, i.e., $f_{i}(x_{i},\cdot)\in\mathcal{H}(\beta,L)$ where $L$ is the same for each $i$.

\end{assumption}

\begin{rem}
  \label{rem:g_Holder}
  Note that the function $g_i$ defined in \eqref{eq:g_i} also belongs to $\mathcal{H}(\beta,L)$ for $i=1,\ldots,p$. Recall that
  \begin{equation*}
    g_{i}(y)=\int x_i f_{i}(x_i,y) dx_i.
  \end{equation*}
  Let $x_i\in\R$ be fixed. If $f(x_i,\cdot)\in\mathcal{H}(\beta,L)$, then we can prove our assertion by direct calculations. We can use a similar argument to prove that $f_Y\in\mathcal{H}(\beta,L)$.
\end{rem}

Denote $\mathcal{F}=\mathcal{F}_{\beta}(L)$ as the class of functions that fulfill Assumption \ref{ass:f_g_Holder}. In the next section, we find the rates of convergence for $\hat{\sigma}_{ij}$ depending on the parameter $\beta$.

To estimate $\hat{\sigma}_{ij}$ in Equation \ref{eq:def_estimator_Sigma_b}, we have to impose compact conditions to the marginal density $f_{Y}$.

\begin{assumption}
  \label{ass:f-small-values}
  Assume that the marginal density of $Y$ is away from 0. Namely that $0<\eta<f_{Y}(y)$ for some given constant $\eta$.
\end{assumption}

In addition, recall that the function $\hat{f}_{Y,b}(y)$ behaves exactly as $\hat{f}_{Y}$ for values greater than $b$. Below of $b$, we set $\hat{f}_{Y}$ at the fixed value $b$.


\begin{rem}
  \label{ass:order_conv_h_b}
  To control the term $\hat{f}_{Y,b}$ in the denominator of equation \eqref{eq:def_estimator_Sigma_b} and ensure the convergence of $\hat{\sigma}_{ij}$, we need to control $h$ and $b$, both values converging to zero. As $n\to\infty$, set $h\sim n^{-c_1}$ and $b\sim n^{-c_2}$, where the positive numbers $c_1$ and $c_2$, which satisfy $c_1/\beta < c_2 < 1/2 -c_1$. The notation ``$\sim$'' means that two quantities have the same order of convergence.
\end{rem}

In the following, we choose a kernel function satisfying the requirements of
Assumption~\ref{ass:kernel} below.

\begin{assumption}
  \label{ass:kernel}
  The continuous kernel function $K(\cdot)$ has order $\lfloor\beta\rfloor$ if it satisfies the following conditions:
  \begin{description}
  \item{$(a)$} the support of $K(\cdot)$ is the interval $[-1,1];$
  \item{$(b)$} $K(\cdot)$ is symmetric around 0;
  \item{$(c)$} $\int_{-1}^{1}K(u)du=1$ and $\int_{-1}^{1}u^{i}K(u)du=0$ for
    $k=1,\ldots,\lfloor\beta\rfloor$.
  \end{description}
\end{assumption}

The most used kernels in the literature are the second-order ones (see \cite{tsybakov2009introduction} and \cite{hardle2004nonparametric} for general references). Also, multiplying second-order kernels by a $(\lfloor\beta\rfloor/2 -1)^{\text{th}}$ order polynomial in $u^{2}$ creates fourth-order kernels. \cite{hansen2005exact} studies explicit constructions for kernels of any order. Table \ref{tab:kernel_examples} presents the order 2 and 4 of the Uniform, Epanechnikov and Biweight kernels.

\begin{table}[htbp]
  \centering
  \begin{tabular}{ll}
    \multicolumn{2}{c}{\textbf{Second-order kernels}} \\[5pt]
    Uniform & $\displaystyle K_{0}(u) = \frac{1}{2}\ I(\vert u \vert \leq 1)$ \\[8pt]
    Epanechnikov & $\displaystyle K_{1}(u) = \frac{3}{4} \ (1-u^{2})I(\vert u \vert \leq 1)$ \\[8pt]
    Biweight & $\displaystyle K_{2}(u) = \frac{15}{16}\ (1-u^{2})^{2} I(\vert u \vert \leq
               1)$ \\[8pt]
    \multicolumn{2}{c}{\textbf{Fourth-order kernels}} \\[5pt]
    Epanechnikov & $\displaystyle K_{4,1}(u) = \frac{15}{8} \left(1-\frac{7}{3}u^{2}\right)\ K_{1}(u)$ \\[8pt]
    Biweight & $\displaystyle K_{4,2}(u) = \frac{7}{4} \left(1-3u^{2}\right)\ K_{2}(u)$ \\[8pt]
  \end{tabular}
  \caption{Examples of kernels satisfying conditions $(a)$, $(b)$ and $(c)$. The function $I(\vert u \vert \leq 1)$ is the indicator function on the interval where $u\in [-1,1]$. }
  \label{tab:kernel_examples}
\end{table}

\subsection{Rate of convergence for the estimates of the matrix
  entries}\label{sec:upper-bound-element}

In this section, we derive the risk upper bound for the elementwise estimator of \eqref{eq:def_Cov_E_X|Y} defined in \eqref{eq:def_estimator_Sigma_b}.

\begin{thm}\label{thm:main_result}
  Assume that $\E\left|X_{i}\right|^{4}<\infty,\ i=1,\ldots,p$. The upper bound risk of the estimator $\hat{\sigma}_{ij}$ defined in \eqref{eq:def_estimator_Sigma_b} over the functional class $\mathcal{F}$ satisfies:
  \begin{equation*}
    \sup_{f\in \mathcal{F}} \E[(\hat{\sigma}_{ij}-\sigma_{ij})^{2}]
    \leq C_1 h^{2\beta}
    + \frac{C_{2}\log^{4}n}{n^{2}h^{4}} + {\frac{C_{3}}{n}}.
  \end{equation*}
  In particular,
  \begin{itemize}
  \item if $\beta\geq 2$ and we choose $n^{-1/4}\leq h \leq
    n^{-1/2\beta}$, then
    \begin{equation}\label{eq:rate_conv_ij_beta_gt_2}
      \sup_{\mathcal{F}} \E[(\hat{\sigma}_{ij}-\sigma_{ij})^{2}]\leq\frac{C}{n};
    \end{equation}

  \item if $\beta<2$ and we choose $h=n^{-1/(\beta + 2)}$, then
    \begin{equation}\label{eq:rate_conv_ij_beta_lt_2}
      \sup_{\mathcal{F}} \E[(\hat{\sigma}_{ij}-\sigma_{ij})^{2}]\leq \left(\frac{\log^2(n)}{n} \right)^{2\beta/(\beta + 2)}.
    \end{equation}
  \end{itemize}

\end{thm}

We postponed the proof for this theorem until the Appendix.

The results in Theorem~\ref{thm:main_result} show that the rate of convergence exhibits an ``elbow'' effect. This means we can recover a parametric rate for regular enough functions $f_Y(y)$ and $g_i(y) \ i=1,\ldots,p$. Otherwise, the mean squared error has a slower rate depending on the regularity of the functional class $\mathcal{F}$. This behaviour is common in functional estimation, e.g., see \cite{donoho1996density}.

Thus, we have provided the rates of convergence for any $\beta>0$ and a $n$-consistency for $\beta\geq2$. In addition, our results agree with those of \cite{zhu1996asymptotics}, who showed the $n$-consistency of the mean squared error by assuming a regularity of $\beta=4$, supporting our method.

In practice, the construction of our estimator requires the choice of the bandwidth $h$. We stress that if the density is smooth enough ($\beta>2$), $h$ can be chosen as $h \sim n^{-1/4}$. In such case, our estimator is adaptive. On the other hand, $h$ should take into acount the value of $\beta$. Ideally, one may need the construction of a adaptation strategy (see, e.g.~\cite{Lepski1997}). However this is not the purpose of the represent contribution.

Now, under some mild conditions, it appears to be natural to investigate
the rate of convergence for the whole matrix estimator
$\hat{\Sigma}=(\hat{\sigma}_{ij})$.
In the next section, we extend the result of Theorem
\ref{thm:main_result} to find the rate of convergence for $\hat{\Sigma}$
under the Frobenius norm.

\subsection{Estimation of the conditional matrix}
\label{sec:consistency-under-frob-norm}

We have got upper bounds for the quadratic risk related to the estimation of each coefficient for the matrix $\Sigma = \C(\E[\boldsymbol{X} \vert Y]) = (\sigma_{ij})_{i,j=1..p}$. Although it is not the main purpose of the paper, we can extend the study to the whole matrix $\Sigma$.


In this paper, the performances of an estimator of $\Sigma$ will be
measured with the Frobenius norm defined as:
\begin{defn}
  Define the Frobenius norm of a matrix $A = (a_{ij})_{p \times p}$ as the
  $\ell^2$ vector norm of all entries in the matrix, where
  \begin{equation*}
    \Vert A \Vert_F^2 = \sum_{i,j} a_{ij}^2.
  \end{equation*}

\end{defn}

The use of other norms like the operator norm, is out of scope for this
work. Basically the technique to control the matrix with this norm is
bounding specific sub-block matrices of $\Sigma$. This development rest
for a future study

Define the estimator $\hat \Sigma= (\hat{\sigma}_{ij})$, where the values
of $\hat{\sigma}_{ij}$ come from Equation
\eqref{eq:def_estimator_Sigma_b}.
According to Theorem \ref{thm:main_result}, we could bound the mean
squared error over the normalized Frobenius norm $\hat{\Sigma}$ and
$\Sigma$ by
\begin{equation*}
  \sup_{f\in\mathcal{F}}\frac{1}{p}\E \Vert \hat{\Sigma}
  -\Sigma \Vert_F^2 \leq \frac{p}{n},
\end{equation*}
when $\beta\geq 2$.

The estimator $\hat{\Sigma}$ turns inconsistent when $p\gg n$; behavior
already explored.
We can cite studies such that: \citet{Muirhead1987},
\citet{johnstone2001distribution}, \citet{bickel2008regularized,
  bickel2008covariance}, and \citet{fan2008high}.

To avoid consistency problems, we establish additional assumptions on the covariance matrix, thus we consider a different version of $\hat{\Sigma}$. We get this modification by setting the coefficients of the matrix to zero from some point. \citet{bickel2008regularized} refer to this transformation as ``banding.''



For an integer $m$ with $1\leq m \leq p$, we define the banding
estimator of $\hat{\Sigma}$ as
\begin{equation}\label{eq:tapering_cov_cond}
  \hat{\Sigma}_m = (w_{ij}\hat{\sigma}_{ij})_{p \times p},
\end{equation}
where \citeauthor{bickel2008regularized} define the function $w_{ij}$ as
\begin{equation*}
  w_{ij}=
  \begin{cases}
    1, & \text{when } \vert i - j \vert \leq m,\\
    0, & \text{otherwise.}
  \end{cases}
\end{equation*}


If $p\gg n$, we require that $\Sigma$ belongs to some space with smooth
decay in its coefficients.
The next assumption fixes $\Sigma$ in a subset of the definite positive
matrices.
\begin{assumption}\label{ass:hypothesis_space_Frob}
  The positive-definite covariance matrix $\Sigma$ belongs to the
  following parameter space:
  \begin{multline}
    \label{eq:G_alpha_decay_cov}
    \mathcal{G}_\alpha %
    = \mathcal{G}_\alpha(M_0,M_{1})\nonumber \\
    = \{ \Sigma : \vert \sigma_{ij} \vert \leq M_1 \vert i - j
    \vert^{-(\alpha + 1)} \text{\ for\ } i \neq j \text{ and }
    \lambda_\text{max}(\Sigma) \leq M_0\},
  \end{multline}
  where $\lambda_\text{max}(\Sigma)$ is the maximum eigenvalue of the
  matrix $\Sigma$, $M_0>0$, and $M_1>0$.
\end{assumption}

\begin{notation}
  \label{not:class_G_prime}
  Set $\mathcal{G}^\prime=\mathcal{G}^\prime_{\alpha,\beta}(L)$ as the
  functional class formed by the intersection between
  $\mathcal{F}_\beta(L)$ and $\mathcal{G}_\alpha$.
\end{notation}

In our case, Assumption \ref{ass:hypothesis_space_Frob} defines a matrix
space indexed by a regularity parameter $\alpha$.
This parameter $\alpha$ defines a rate of decay for the coefficients of
the conditional covariance as they move away from the diagonal %
\citet{bickel2008regularized} and \citet{cai2010optimal} discuss in
detail the assumption.

The space $\mathcal{G}_\alpha$ depends on the inherent matrix structure
of the problem.
In the literature we can find alternatives types of matrix structures
and spaces, for example: \cite{cai2013toeplitz} studied a Toeplitz
structure for the covariance matrix; \cite{xiao2014theoretic} worked
with a similar matrix structure as ours, but with an alternative way to
decrease the coefficients out of the diagonal; and \cite{cai2012sparse}
used a sparse configuration to estimate optimal rates for the estimated
covariance.

Also, other authors generalize the banding technique used before.
We can use thresholding which removes the ``small'' coefficients of the
matrix (see \cite{cai2012adaptive}), tapering which decreases the
elements off the diagonal with a linear function (see
\cite{cai2010optimal}), or a nested lasso which penalize the
coefficients according their assumptions (see \cite{levina2008sparse}),
to cite some examples.

The following theorem provides an upper bound of the convergence rate
for the estimate defined in \eqref{eq:tapering_cov_cond} under the
normalized Frobenius norm based on the sample $\{(\boldsymbol{X}_1,Y_1),
\ldots, (\boldsymbol{X}_n,Y_n)\}$.

\begin{thm}\label{thm:conv-frob-norm}
  Assume that $\E\left|X_{i}\right|^{4}<\infty,\ k=1,\ldots,p$.
  The estimator $\hat{\Sigma} = (\hat{\sigma}_{ij})$ defined in
  \eqref{eq:def_estimator_Sigma_b} over the functional class
  $\mathcal{G}^\prime$ satisfies the following.

  \begin{itemize}
  \item If $\beta \geq 2$,
    \begin{equation}
      \label{eq:thm2_beta_greater_2}
      \sup_{\mathcal{G}^\prime }\frac{1}{p}\E \Vert \hat{\Sigma}_m
      -\Sigma \Vert_F^2 \leq \min \left\{n^{-\textstyle\frac{2\alpha +
            1}{2(\alpha + 1)} }, \frac{p}{n} \right\},
    \end{equation}
    where $m=p$ if $n^{1/(2(\alpha+1))}>p$ or $m=n^{1/(2(\alpha+1))}$
    otherwise.
  \item If $\beta < 2$,
    \begin{multline}
      \label{eq:thm2_beta_less_2}
      \sup_{\mathcal{G}^\prime } \frac{1}{p} \E \Vert \hat{\Sigma}_m
      -\Sigma \Vert_F \leq \min \left\{ \left(\frac{\log^2n} {n}
        \right)^{ \textstyle\frac{2\beta(2\alpha+1)} {(\beta +
            2)(2(\alpha + 1))}}, p\left(\frac{\log^2
            n}{n}\right)^{\textstyle\frac{2\beta}{\beta + 2}}\right\},
    \end{multline}
    where $m=p$ if
    $(\log^2n/n)^{-2\beta/(2(\alpha+1)(\beta+2))}>p$ or \\
    $m=(\log^2n/n)^{-2\beta/(2(\alpha+1)(\beta+2))}$ otherwise.
  \end{itemize}
\end{thm}

The minimum reached in Equations \eqref{eq:thm2_beta_greater_2} and
\eqref{eq:thm2_beta_less_2} differs according to the value of the ratio
$p/n$.
If $p\ll n$, we use the original covariance matrix $\hat{\Sigma}$.
Otherwise, it is necessary to regularize the estimator to maintain the
consistency.
For example, in the case where $\beta\geq 2$, if
$n^{1/(2(\alpha+1))}>p$, then we are in a low-dimensional framework and
we use the full matrix $\hat{\Sigma}$.
In other cases, when $n^{1/(2(\alpha+1))}\leq p$, the regularization of
the matrix is mandatory and we choose $m=n^{1/(2(\alpha+1))}$ to
generate the matrix $\hat{\Sigma}_m$ defined in
\eqref{eq:tapering_cov_cond}.
A similar analysis can be performed if $\beta < 2$.

\section{Sliced Inverse Regression for time series}
\label{sec:sir-time-series}

In the Introduction we mentioned the link between the SIR method and time series process. The auto-regressive time series have an autocovariance matrix which decreases if the lag between times gets larger. Therefore, events that occurred in far in the past, do not affect the outcome in the present. This behavior fit the choose of the banding estimator in \eqref{eq:tapering_cov_cond}.

The rest of the section is to present the general framework to use SIR for time series models.

Consider a multivariate time series $\{\bm{X}_{t} = (X_{1,t}, \ldots, X_{m,t})\colon t=1,\ldots, T\}$. The classic vector auto-regressive with $p$ lags (VAR(p)) describes how to interact the past and the present of certain variable, through some linear model. We write this models as
\begin{equation*}
  \bm{X}_{t} = \bm{\Phi}_{1}\bm{X}_{t-1} + \ldots + \bm{\Phi}_{p}\bm{X}_{t-p} + \bm{\varepsilon}_{t}
\end{equation*}
where $\bm{\Phi}_{1}, \ldots, \bm{\Phi}_{p}\in \mathbb{R}^{m\times m}$ are
matrices and $\bm{\varepsilon}_{t}$ is white noise process.

However, applications in economy, medicine or social sciences need
complex (and sometimes unknown) structures to model the time series.
According to \cite{Tong1993}, we can write a multivariate nonlinear
autoregresive model of order $p$ as
\begin{equation*}
  \bm{X}_{t} = h(\bm{X}_{t-1}, \ldots, \bm{X}_{t-p}, \varepsilon_{t}).
\end{equation*}

Here, the function $h:\R^{(p+1)m} \to \R^{m}$ relates the past and present
of the process through a nonlinear process. One key assumption is that
the influence of $Y$ does not affect the behavior of $\bm{X}$ and
$\varepsilon_{t}$ is uncorrelated with all past explanatory variables
(i.e. $\bm{X}_{j}$ with $j\leq t-1$).

\citet{Becker2002} studied this process by focusing on the influence of
$\bm{X}$ to $Y$. They suggest a model (possibly nonlinear) with a link
function
\begin{equation*}
  Y_{t} = F(\bm{X}_{t}, Y_{t-1}, \bm{X}_{t-1}, \ldots, Y_{t-p},
  \bm{X}_{t-p}, \varepsilon_{t})
\end{equation*}
where $\varepsilon_{t}$ is again a white noise.

In the SIR framework, $(\bm{X}_{t}, Y_{t})$ correspond to the observation
$(\bm{X}, Y)$ in Equation \eqref{eq:sir}. To adapt this problem to the
SIR framework, \citeauthor{Becker2002} suggest to search the effective
dimension directions using a modified version of Equation
\eqref{eq:sir-simplified},
\begin{equation*}
  \label{eq:Y-VAR}
  Y_{t} = \varphi\left( \nu_{1}^{\top} \widetilde{\boldsymbol{X}}_{t},
    \ldots, \nu_{k}^{\top}\widetilde{\boldsymbol{X}}_{t}, \varepsilon
  \right)\quad \forall t,
\end{equation*}
where $\widetilde{\boldsymbol{X}}_{t} = (\bm{X}_{t}, Y_{t-1}, \bm{X}_{t-1},
\ldots, Y_{t-p}, \bm{X}_{t-p})$ and $k\ll p$.

Therefore, we apply the SIR method to a the sample
$(\widetilde{\boldsymbol{X}}_{t}, Y_{t})$ with a number $p$ of lags
chosen empirically. Specifically, the main goal in this case is to
estimate $\C(\E[\widetilde{\boldsymbol{X}}_{t} \vert Y_{t}])$.

This procedure finds the minimun number of lags where the model could
capture all the relevant information.

Theorem \ref{thm:conv-frob-norm} claims that if $F$ is H\"older with
$\beta\geq 2$ and other conditions, then the nonparametric estimator
\eqref{eq:tapering_cov_cond} converge to $\C(\E[\boldsymbol{X}\vert Y])$
asymptotically at speed $n^{-1}$. As a consequence, the eigenvectors of
$\hat{\Sigma}_m$ correctly estimate the effective dimension reduction
directions. Thus, we have an accurate characterization of the effective
dimension reduction space spanned by the eigenvectors of
$\hat{\Sigma}_m$. Equation \ref{eq:tapering_cov_cond} estimates the
$\C(\E[\bm{X}\vert Y])$ assuming that the covariance belongs to a class
$\mathcal{F}$ with decreasing coefficients which fits in this framework
given the model's autoregresive nature.

\section{Conclusion}\label{sec:Conclusions}

In this study, we investigated the convergence rate of a nonparametric estimator for the conditional covariance $\C(\E[X \vert Y])$. First, we studied the nonparametric behavior of each element of the matrix based on the study of \citet{zhu1996asymptotics}. This approach allowed us to exhibit rates of convergence according to the smoothness of the target. In particular we showed that if the model was regular enough, it achieved a parametric rate of $1/n$. Otherwise, we got a slower rate $(\log^2(n)/n)^{2\beta/(\beta+2)}$, depending on the regularity, $\beta$, of a certain H\"{o}lder class.

As a natural extension, we studied how performs the mean squared risk of $\hat{\Sigma}$ under the Frobenius norm. To maintain consistency and avoid issues due to the high dimensionality of data, the matrix $\hat{\Sigma}$ had to be regularized. We used a regularized version of $\hat{\Sigma}$ called $\hat{\Sigma}_m$, which we got by a Schur multiplication between $\hat{\Sigma}$ and a positive definite matrix of weights. These weights remained 1 until some point away from the diagonal and zero elsewhere.

This method could not make sure the positive definiteness of the estimate. We proved that under some mild conditions, our estimator was consistent given that either $p/n$ or $p(\log^2(n)/n)^{2\beta/(\beta+2)}$ zeroed, so the estimator would be positive definite with a probability tending to 1. \citet{cai2010optimal} suggested project $\hat{\Sigma}$ into the space of positive-semidefinite matrices under the operator norm. In simple terms, first diagonalizing $\hat{\Sigma}$ and then replacing the negative eigenvalues by 0, where the matrix got is semidefinite positive.

Other norms are available to measure the error in matrix estimation, such as the operator norm.
However, this approach requires concentration inequalities for some specific matrices blocks of $\hat{\Sigma}$, which are not readily available. However, the operator norm in the nonparametric estimation of $\Sigma$ is worth investigating.

One application to time series was presented, showing the model capabilities. The immediate pending task will be to implement these ideas in a programming language and show its performance.



\section{Appendix}\label{sec:Appendix}

\subsection{Technical lemmas}

\begin{lem}\label{lem:Bias}
  Let $\B(\hat{\sigma}_{ij}) = \vert \E(\hat{\sigma}_{ij} - \sigma_{ij})\vert$. Then, under the same assumptions as Theorem \ref{thm:main_result} and supposing that $nh\to 0$ as $n\to \infty$, we have

   \begin{equation*}
     \B^{2}(\hat{\sigma}_{ij})\leq C_{1}h^{2\beta}
     + \frac{C_{2}}{n^{2}h^{2}} + \frac{1}{n^2}
   \end{equation*}

   for $C_{1}$ and $C_{2}$, which are positive constants that depend only
   on $L$, $s$, $\beta$, and the kernel $K$.
 \end{lem}

 \begin{lem}\label{lem:Var}
   Under the same assumptions as Theorem \ref{thm:main_result}, if we
   suppose that $nh\to 0$ as $n\to \infty$, then we have

  \begin{equation*}
    \V\left(\hat{\sigma}_{ij}\right)\leq
    C_1 h^{2\beta} + \frac{C_{2}\log^{4}n}{n^{2}h^{4}}
    + \frac{1}{n}
  \end{equation*}

  for $C_{1}$, $C_{2}$, and $C_{3}$, which are positive constants that
  depend only on $L$, $s$, $\beta$, and the kernel $K$.
\end{lem}

\begin{lem}[cf. \citet{prakasa1983nonparametric}, Theorem
 2.1.8]\label{lem:f_hat_conv_rate}
 Suppose that $K$ is a kernel of order
 $s=\left\lfloor\beta\right\rfloor$ and that Assumption
 \ref{ass:f_g_Holder} is satisfied.
 Then,
 \begin{equation*}
   \sup_{y}\vert\hat{f}(y)-f(y)\vert=O\left(h^{\beta} + \frac{\log
       n}{n^{1/2}h}\right) \quad a.s.
 \end{equation*}
\end{lem}

The following lemma is a modified version of Theorem 2.37 from
\citet{pollard1984convergence}.

\begin{lem}\label{lem:g_hat_conv_rate}
  Suppose that $K$ is a kernel of order
  $s=\left\lfloor\beta\right\rfloor$, $\E[X_i^4]<\infty$ and that
  Assumption \ref{ass:f_g_Holder} is satisfied.
  For any $\varepsilon>0$,
  \begin{multline*}
    \P\left(\sup_{y}\vert\hat{g}_{i}(y) -
      \E\left[\hat{g}_{i}(y)\right]\vert > 8n^{-1/2}h^{-1}\varepsilon\right) \\
    \leq 2 c\left(\frac{\varepsilon}{\sqrt{n}d}\right)^{-4}\exp\left\{
      -\frac{1}{2}\varepsilon^{2}/\left(32d(\log n)^{1/2}\right)\right\} \\
    + 8cd^{-8}\exp\left(-nd^{2}\right) +
    \E\left[X_{i}^{4}\right]I\left(\vert X_{i}\vert>cd^{-1/2}(\log
      n)^{1/4}\right),
  \end{multline*}
  where
  \[
    d\geq\sup_{y}\left\{
      \V\left(K\left(\frac{y-Y}{h}\right)\right)\right\} ^{1/2}
  \].

\end{lem}
We refer to \citet{zhu1996asymptotics} for the proof of Lemma
\ref{lem:g_hat_conv_rate}.
Using the last result, the uniform convergence rate of $\hat{g}_{i}(y)$
can be obtained.

\begin{lem}\label{lem:bound_g_hat_minus_g}
  Suppose that $K$ is a kernel of order
  $s=\left\lfloor\beta\right\rfloor$, $\E[X_i^4]<\infty$ and that
  Assumption \ref{ass:f_g_Holder} is satisfied.
  Then,
  \[
    \sup_{y}\vert\hat{g}_{i}(y)-g_{i}(y)\vert = O_{p}\left(h^{\beta} +
      \frac{\log n}{n^{1/2}h}\right).
  \]
\end{lem}

\subsection{Proofs of Theorems and Lemmas}

\begin{proof}[Proof of Theorem \ref{thm:main_result}]
  First, we consider the usual bias-variance decomposition.
  \begin{equation*}\label{eq:MSE_decomposition}
    \E\left[(\hat{\sigma}_{ij} - \sigma_{ij})^{2}\right] = \B^{2}(\hat{\sigma}_{ij}) + \V(\hat{\sigma}_{ij}),
  \end{equation*}
  where $\B(\hat{\sigma}_{ij}) = \E[\hat{\sigma}_{ij}] - \sigma_{ij}$ and
  $\V(\hat{\sigma}_{ij})=\E[\hat{\sigma}_{ij}^2] -
  \E[\hat{\sigma}_{ij}]^2$.



  By Lemmas \ref{lem:Bias} and \ref{lem:Var}, we obtain the following
  upper bound for the estimation error
  \begin{equation*} \label{eq:final_bound_MSE}
    \E[(\hat{\sigma}_{ij}-\sigma_{ij})^{2}] \leq C_1 h^{2\beta} +
    \frac{C_{2}\log^{4}n}{n^{2}h^{4}} + \frac{1}{n}.
  \end{equation*}

  Depending on the regularity of the model, we consider two cases as
  follows.

  \begin{itemize}
  \item If $\beta \geq 2$, then we can choose $h$ such that
    \[
      \frac{1}{n^{1/4}}\leq h \leq \frac{1}{n^{1/2\beta}},
    \]
    and thus
    \[
      h^{2\beta} \leq \frac{1}{n},\quad \frac{1}{n^2 h^4} \leq \frac{1}{n},
    \]
    thereby concluding the result.

  \item Otherwise, if $\beta<2$, we need to find $h$ such that
    \[
      h = \argmin_h \left( h^{2\beta} + \frac{\log^4(n)}{n^{2}h^{4}}\right).
    \]
    We obtain $h=(\log^2(n)/n)^{1/(\beta + 2)}$ and the risk is bounded by
    \[
      \sup_{\mathcal{F}} \E[(\hat{\sigma}_{ij}-\sigma_{ij})^{2}]\leq \left(
        \frac{\log^2(n)}{n} \right)^{2\beta/(\beta + 2)}.
    \]
  \end{itemize}

\end{proof}

\begin{proof}[Proof of Theorem \ref{thm:conv-frob-norm}]

  For the estimator \eqref{eq:tapering_cov_cond}, we have
  \begin{equation*}
    \E\Vert \hat{\Sigma}_m -\Sigma \Vert_F^2
    = \sum_{i,j=1}^p \E(w_{ij} \hat{\sigma}_{ij} - \sigma_{ij})^2
  \end{equation*}
  Let $i,j\in\{1,\ldots,p\}$ be fixed.
  Then,
  \begin{align*}
    \E (w_{ij} \hat{\sigma}_{ij} - \sigma_{ij} )^2 %
    & = w^2_{ij} \E[(\hat{\sigma}_{ij} - \sigma)^2 ]
      + ( 1 - w_{ij} )^2 \sigma^2_{ij} \\
    & \leq w^2_{ij} \gamma_{n_{1}} + ( 1 - w_{ij} )^2 \sigma^2_{ij},
  \end{align*}
  where $\gamma_{n} $is the rate \eqref{eq:rate_conv_ij_beta_gt_2} or
  \eqref{eq:rate_conv_ij_beta_lt_2} depending on the value of $\beta$.
  Furthermore,
  \begin{align*}
    \frac{1}{p} \E \Vert \hat{\Sigma} - \Sigma \Vert_F %
    & \leq \frac{1}{p}
      \sum_{\{(i,j) \colon \vert i - j\vert > m \}} \sigma^2_{ij}
      + \frac{1}{p}
      \sum_{\{ (i,j) \colon \vert i-j\vert \leq m\}} \gamma_{n} \\
    & \equiv R_1 + R_2.
  \end{align*}

  The cardinality of $\{ (i,j) \colon \vert i-j\vert \leq m\}$ is
  bounded by $mp$, so we directly find that ${R_2 \leq C m
    \gamma_{n}}$.


  Thus, using Assumption \ref{ass:hypothesis_space_Frob}, we show that

 \begin{equation*}
   \sup_{\mathcal{G}^\prime } \frac{1}{p} \sum_{\{(i,j) \colon \vert i-j \vert
     > m\}} \sigma^2_{ij} \leq C m^{-2\alpha-1},
 \end{equation*}

 where $\vert \sigma_{ij} \vert \leq C_1 \vert i-j \vert^{-(\alpha+1)}$
 for all $ j\neq i$.
 Thus,
 \begin{equation}\label{eq:risk_m_less_p}
   \sup_{\mathcal{G}^\prime}
   \frac{1}{p} \E \Vert \hat{\Sigma} -\Sigma \Vert_F^2 \leq C m^{-2\alpha
     -1} + C m \gamma_{n} \leq C_2 \gamma_{n}^{(2\alpha +
     1)/(2(\alpha + 1))}
 \end{equation}
 by choosing
 \[
   m = \gamma_{n}^{-1/(2(\alpha + 1))}
 \]
 if $\gamma_{n}^{-1/(2(\alpha + 1))}\leq p$.
 In the case of $\gamma_{n}^{-1/(2(\alpha + 1))}>p$, we choose $m=p$, so
 the bias part is 0 and consequently
 \begin{equation}\label{eq:risk_m_greater_p}
   \frac{1}{p} \E \Vert \hat{\Sigma} -\Sigma \Vert_F \leq C m
   \gamma_{n}.
 \end{equation}
 Using the result of Theorem \ref{thm:main_result}, we distinguish two
 cases depending on the regularity of the model.
 If $\beta \geq 2$, then we take $\gamma_{n} = 1/n$, and if $\beta<2$,
 then $\gamma_{n} = (\log^2 n/n)^{ 2\beta / ( \beta + 2)}$.
 The result is obtained by combining the latter with
 \eqref{eq:risk_m_less_p} and \eqref{eq:risk_m_greater_p}.
  %
  %

\end{proof}

\begin{proof}[Proof of Lemma \ref{lem:Bias}]

  The proof comprises three steps.

  \paragraph{Step 1:} Prove the following.

 \begin{multline*}
   \Big\vert \E[\hat{\sigma}_{ij}] - \sigma_{ij} \Big\vert %
   \leq %
   \left\vert \frac{1}{(n-1)^{2}} %
     \E \left[ \frac{1}{f_{Y}^{2}(Y_{1})} \sum_{k=2}^{n} X_{ik} X_{jk} %
       K_{h}^{2} (Y_{1}-Y_{k}) \right] \right\vert \\
   + \left\vert \frac{1}{(n-1)^{2}} %
     \E \left[ \frac{1}{f_{Y}^{2}(Y_{1})} \sum_{\substack{k,r=2\\ k \neq
           r}}^{n} %
       X_{ik} X_{jr} K_{h} (Y_{1}-Y_{k}) K_{h}
       (Y_{1}-Y_{r}) \right] - \sigma_{ij}\right\vert \\
   + \left\vert %
     \E\left[ \frac{\hat{g}_{i}(Y_1) \hat{g}_{j}(Y_1)} %
       {f_{Y}^{2}(Y_1)} \right] %
     \left( 2\frac{(f_{Y}(y) - \hat{f}_{Y,b}(y) )}{\hat{f}_{Y,b}(y)} +
       \frac{(f_{Y}(y) - \hat{f}_{Y,b}(y))^2}{\hat{f}^2_{Y,b}(y)} %
     \right) \right\vert
 \end{multline*}


 Note that $\hat{g}_{i}(Y_{k}) \hat{g}_{j}(Y_{k})/ \hat{f}_{Y,b}^{2}(Y_{k})$
 are dependent random variables for $k=1,\ldots,n$, which have the same
 distribution.
 Thus,
 \begin{equation*}
   \E [\hat{\sigma}_{ij}] %
   = \E\left[ %
     \frac{1}{n} \sum_{k=1}^{n}
     \frac{ \hat{g}_{i}(Y_{k})\hat{g}_{j}(Y_{k})} %
     {\hat{f}_{Y,b}^{2}(Y_{k})}
   \right]
   = \E\left[
     \frac{\hat{g}_{i}(Y_1) \hat{g}_{j}(Y_1)} %
     { \hat{f}_{Y,b}^{2}(Y_1)}
   \right].
 \end{equation*}

 Furthermore, note that
 \begin{equation} %
   \label{eq:f_div_fhat}
   \frac{f_{Y}^{2}(y)}{\hat{f}_{Y,b}^{2}(y)} %
   = 1 + 2\frac{(f_{Y}(y) - \hat{f}_{Y,b}(y) )}{\hat{f}_{Y,b}(y)} %
   + \frac{(f_{Y}(y) - \hat{f}_{Y,b}(y) )^2}{\hat{f}^2_{Y,b}(y)}.
 \end{equation}

 Denote
 \begin{equation*}
   B = \frac{\hat{g}_{i}(Y_1) \hat{g}_{j}(Y_1)}{f^{2}(Y_1)}.
 \end{equation*}

 By developing the last equation, we obtain
 \begin{align}
   \label{eq:B1_plus_B2}
   B %
   & =\frac{1}{(n-1)^{2}} \frac{1}{f_{Y}^{2}(Y_{1})} \left(\sum_{k=2}^{n}
     X_{ik} K_{h}(Y_{1}-Y_{k}) \right)
     \left(\sum_{k=2}^{n} X_{jk} K_{h} (Y_{1}-Y_{l}) \right) \nonumber \\
   & =\frac{1}{(n-1)^{2}} \frac{1}{f_{Y}^{2}(Y_{1})}
     \sum_{k=2}^{n} X_{ik} X_{jk} K_{h}^{2} (Y_{1}-Y_{k}) \nonumber\\
   & \qquad + \frac{1}{(n-1)^{2}} \frac{1}{f_{Y}^{2}(Y_{1})}
     \sum_{\substack{k,r=2\\ k \neq r}}^{n} X_{ik} X_{jr} K_{h} (Y_{1}-Y_{k})
   K_{h} (Y_{1}-Y_{r}) \nonumber \\
   & \equiv B_{1} + B_{2}  .
 \end{align}

 Finally, if we multiply and divide by $f_{Y}^2(Y)$ inside the
 expectation, and apply Equations \eqref{eq:f_div_fhat} and
 \eqref{eq:B1_plus_B2}, we obtain the following.
 \begin{align*}
   \label{eq:g_i_g_j_div_f_Y1}
   \Big\vert \E[\hat{\sigma}_{ij}] - \sigma_{ij} \Big\vert %
   & = \Bigg\vert \E\left[\frac{\hat{g}_{i}(Y_1)
     \hat{g}_{j}(Y_1)}{\hat{f}_{Y,b}^{2}(Y_1)}\right] - \sigma_{ij} \Bigg\vert \\
   & = \Bigg\vert \E \left[ \frac{\hat{g}_{i}(Y_1) \hat{g}_{j}(Y_1)}
     {f_{Y}^{2}(Y_1)} \frac{f_{Y}^{2}(Y_{1})} {\hat{f}_{Y,b}^{2}(Y_{1})}
     \right]
     - \sigma_{ij} \Bigg\vert \\
   & \leq \Big\vert \E[B_{1}] \Big\vert + \Big\vert \E[B_{2}] - \sigma_{ij} \Big\vert \\
   & \qquad + \Bigg\vert \E\bigg[B \bigg( 2\frac{(f_{Y}(y) - \hat{f}_{Y,b}(y)
     )}{\hat{f}_{Y,b}(y)} + \frac{(f_{Y}(y) -
     \hat{f}_{Y,b}(y))^2}{\hat{f}^2_{Y,b}(y)}
     \bigg) \bigg] \Bigg\vert
 \end{align*}





 \paragraph{Step 2:}
 Prove the following.

 \begin{equation*}
   \Big\vert \E[B_{1}] \Big\vert \leq \frac{C}{nh}
 \end{equation*}

 By conditioning with respect to $Y_1$ and noting that the random
 variables inside have the same distribution, we compute

 \begin{align*}
   \E[B_{1}] %
   & = \frac{1}{(n-1)^{2}} \E \left[ \frac{1}{f_{Y}^{2}(Y_{1})} \E \left[
     \left.\sum_{k=2}^{n} X_{ik} X_{jk} K_{h}^{2} (Y_{1}-Y_{k}) \
     \right|\ Y_{1}\right]\right] \\
   & = \frac{1}{n-1} \E \left[\frac{X_{i2}
     X_{j2} K_{h}^{2} (Y_{1}-Y_{2})} {f_{Y}^{2}(Y_{1})} \right].
 \end{align*}

 By writing the last expression based on its integral form and
 remembering that $g_{i}(y) = \int x_{i}f(x_{i},y)dx_{i}$, we have the
 following.

 \begin{align*}
   \E[B_{1}] %
   & = \frac{1}{n-1} \int \left( \int x_{i} K_{h}(y_{1}-y_{2}) %
     f(x_{i},y) dx_{i}dy_{2}\right) \\
   & \qquad \left( \int x_{j} K_{h}(y_{1}-y_{2}) %
     f(x_{j},y) dx_{j}dy_{2}\right) \frac{1}{f_{Y}(y_{1})} dy_{1} \\
   & = \frac{1}{n-1} \int \frac{g_{i}(y_{2}) g_{j}(y_{2})}{f_{Y}(y)} %
     K_{h}^{2}(y_{1}-y_{2}) dy_{1} dy_{2} \\
   & = \frac{1}{(n-1)h} \int \frac{g_{i}(y_{1}+uh) g_{j}(y_{1}+uh)}{f_{Y}(y)} %
     K^{2}(u) du dy_{1} \\
 \end{align*}

 We remark that given $a_m$ and $b_m,\ m=1,2$ are real numbers such as
 \linebreak $a_m~<~b_m$, then the integrals containing the coordinate
 $(x,y)$ will be evaluated in the cube $[a_1,b_1]\times[a_2,b_2]$.

 Define the supremum norm of $f$ as $\Vert f \Vert_\infty =
 \sup\{f(x,y)\in [a_1,b_1]\times[a_2,b_2] \}$.
 As a consequence, we can also define the supremum norm of $g_{i}$ as
 $\Vert g_{i} \Vert_\infty = \sup\{g_{i}(y)\in [b_{1},b_{2}]\}$.
 Therefore,

 \begin{equation*}
   \Big\vert \E[B_{1}] \Big\vert \leq
   \frac{\Vert g_{i}\Vert_{\infty} \Vert g_{j}\Vert_{\infty} }{(n-1)h}
   \int \frac{1}{f_{Y}(y)} K^{2}(u) du dy_{1},
 \end{equation*}

 which leads to

 \begin{equation*}
   \Big\vert \E[B_{1}] \Big\vert \leq \frac{C}{nh}.
 \end{equation*}

 \paragraph{Step 3:} Prove the following.

 \begin{multline*}
   \Bigg\vert \E\bigg[B \bigg( 2\frac{(f_{Y}(y) - \hat{f}_{Y,b}(y)
     )}{\hat{f}_{Y,b}(y)} + \frac{(f_{Y}(y) -
     \hat{f}_{Y,b}(y))^2}{\hat{f}^2_{Y,b}(y)} \bigg) \bigg] \Bigg\vert
   \\ %
   \leq C (n^{c_1-\beta c_2} + n^{1/2 -c_1 -c_2}\log n)
 \end{multline*}

 By Lemma \ref{lem:f_hat_conv_rate}, we have

 \begin{align*}
   & \Bigg\vert \E\bigg[B \bigg( 2\frac{(f_{Y}(y) - \hat{f}_{Y,b}(y)
     )}{\hat{f}_{Y,b}(y)} + \frac{(f_{Y}(y) -
     \hat{f}_{Y,b}(y))^2}{\hat{f}^2_{Y,b}(y)}
     \bigg) \bigg] \Bigg\vert \\
   & \leq \left \vert \E[B]
     \right \vert \left( C_1 n^{c_1}\left(h^\beta
     + \frac{\log n}{n^{1/2}h}\right)
     + C_2 n^{2c_1}\left(h^\beta + \frac{\log
     n}{n^{1/2}h}\right)^2 \right) \\
   & \leq C (n^{c_1-\beta c_2} + n^{1/2 -c_1 -c_2}\log n).
 \end{align*}

 The last line is given by the assumption that $\E[X^4]$ is finite and by
 Remark \ref{ass:order_conv_h_b}.

 \paragraph{Step 4:} Show that

\begin{equation*}
  \vert \E[B_2] - \sigma_{ij} \vert \leq C h^{2\beta}.
\end{equation*}

This term can be bounded as follows

\begin{multline}\label{eq:B_2}
  \E[B_2] = \left( \frac{n-2}{n-1} \right) \E \left( \frac{X_{i2} X_{j3}
      K_{h} (Y_{1} - Y_{2}) K_{h} (Y_{1} - Y_{3})}
    {f_{Y}^{2} (Y_{1})} \right) \\
  \leq \int \left (\int x_{i} K_{h} (y_{1} - y) f(x,y) dx_{i} dy
  \right) \\
  \hspace{2em} \left( \int x_{j} K_{h} (y_{1} - y) f(x,y) dx_{j} dy \right)
  \frac{1}{f_{Y}(y_{1})} dy_{1}.
\end{multline}
By Assumption \ref{ass:f_g_Holder} with $s=\lfloor\beta\rfloor$ and for
$0<\tau<1$, we have
\begin{align*}
  & \hspace{-2em} \int K_{h} (y_1 - y) f(x,y) dy - f(x,y) \\
  & = \int K(u) f(x,uh + y) du - f(x,y) \\
  & =\int K(u) \left\{ f(x,y) + uh f^{\prime}(x,y) + \cdots +
    \frac{f^{(s)}(x,y + \tau uh)}{s!} \left( uh
    \right)^{j} \right\} du \\
  & =\frac{1}{s!} \int K(u)\left( uh \right)^{s} f^{(s)} (x,y + \tau uh) du.
\end{align*}
By adding $K(u)(uh)^sf^{(s)}(x,y)$ to the last integral and given that
$f(x,\cdot)\in \mathcal{H}(\beta,L)$, we can see that
\begin{align*}
  & \hspace{-2em} \int K_{h}(y_1-y) f(x,y)dy - f(x,y) \\
  & = \frac{1}{s!} \int K(u) \left( uh \right)^{s} \left\{ f^{(s)}
    (x,y + \tau uh)-f^{(s)}(x,y) \right\} du \\
  & \leq \frac{1}{s!} \int K(u) \left( uh \right)^{s} \left( \tau
    uh\right)^{\beta - s} du \\
  & \leq \left( \frac{1}{s!} \int \vert u^{\beta} K(u) \vert du \right)\tau
    h^{\beta} = Ch^{\beta}.
\end{align*}


By plugging this into \eqref{eq:B_2}, we obtain
\begin{align*}
  \vert \E[B_2] - \sigma_{ij} \vert %
  & \leq \Bigg\vert \int \left\{ \left(\int x_{i}
    \left( Ch^{\beta} +
    f(x_{i},y) \right) dx_{i} \right) \left(\int x_{j}
    \left( Ch^{\beta} +
    f(x_{j},y) \right) dx_{j} \right) \right.\\
  & \left.\qquad - \left(\int x_{i} f(x_{i},y) dx_{i} \right)
    \left(\int
    x_{j} f(x_{j},y) dx_{j} \right) \right\} \frac{1}{f_{Y}(y)} dy
    \Bigg\vert \\
  & \leq C^2 h^{2\beta} \Bigg\vert\int \frac{x_i x_j}{f_{Y}(y)} \ dx_i dx_j
    dy \Bigg\vert \\
  & \hspace{2em} + C h^\beta
    \Bigg\vert\left(
    \int x_i x_j f(x_i, y)\ dx_j dx_i dy
    + \int x_i x_j f(x_j, y)\ dx_j dx_i dy
    \right)\Bigg\vert \\
  & \leq C_{1} h^{\beta} + C_{2}h^{2\beta} \\
  & \leq C_{1} h^{\beta}.
\end{align*}

%
\paragraph{Step 5:} By combining the results from Steps 1 to 4, we have

 \begin{equation*}
   \B(\hat{\sigma}_{ij}) \leq C_{1} h^{\beta} + \frac{C_{2}}{nh}
   + C_3 (n^{c_1-\beta c_2} + n^{1/2 -c_1 -c_2}\log n).
 \end{equation*}

 which allows us to conclude that

\begin{equation*}
  \B^{2}( \hat{\sigma}_{ij})\leq C_{1} h^{2\beta}
  + \frac{C_{2}}{n^{2} h^{2}}.
\end{equation*}
\end{proof}

\begin{proof}[Proof of Lemma \ref{lem:Var}]

  The proof comprises several steps.
  Define
  \begin{align*}
    R_{i,b}(Y) & = \frac{g_{i}(Y)}{f_{Y}(Y)},\\
    V_1(Y) & = \frac{g_{i}(Y) g_{j}(Y)}{f^{2}_{Y,b}(Y)}
            = R_{i,b}R_{j,b} , \\
    V_2(Y) & = \frac{g_{i}(Y)}{f^{2}_{Y,b}(Y)} \left(\hat{g}_{j}(Y) -
            g_{j}(Y)\right)
            = \frac{R_{i,b}}{f_{Y}(Y)} \left(\hat{g}_{j}(Y) -
            g_{j}(Y)\right), \\
    V_3(Y) & = \frac{g_{j}(Y)}{f^{2}_{Y,b}(Y)} \left(\hat{g}_{i}(Y) -
            g_{i}(Y)\right)
            = \frac{R_{j,b}}{f_{Y}(Y)} \left(\hat{g}_{i}(Y) -
            g_{i}(Y)\right), \\
    V_4(Y) & = \frac{1}{f^{2}_{Y,b}(Y)} \left(\hat{g}_{i}(Y) - g_{i}(Y)\right)
            \left(\hat{g}_{j}(Y) - g_{j}(Y)\right), \\
    J_{n}(Y) & = (V_1(Y) + V_2(Y) + V_3(Y) + V_4(Y))\\
              & \hspace{4em} \left(2\frac{(f_{Y}(Y) -
                \hat{f}_{Y,b}(Y))}{\hat{f}_{Y,b}(Y)}
                + \frac{(f_{Y}(Y) -
                \hat{f}_{Y,b}(Y))^2}{\hat{f}^2_{Y,b}(Y)}\right).
  \end{align*}
  It is clear that $\hat{\sigma}_{ij} = n^{-1} \sum_{k=1}^{n} V_1(Y_k) +
  V_2(Y_k) + V_3(Y_k) + V_4(Y_k) + J_n(Y_k)$.
  If $C>0$, then the variance $\V(\hat{\sigma}_{ij})$ is bounded by
  \begin{multline*}\label{eq:Variance_decomposition}
    C \left\{\V\left(\frac{1}{n}\sum_{k=1}^{n}V_{1}(Y_k)\right) +
      \V\left(\frac{1}{n}\sum_{k=1}^{n}V_{2}(Y_k)\right) +
      \V\left(\frac{1}{n}\sum_{k=1}^{n}V_{3}(Y_k)\right)\right.\\
    + \left.\V\left(\frac{1}{n}\sum_{k=1}^{n}V_{4}(Y_k)\right) +
      \V\left(\frac{1}{n}\sum_{k=1}^{n}J_n(Y_k)\right)\right\}.
  \end{multline*}
  We bound each term separately.

  \paragraph{Step 1:}

  Prove that

 \[
   \V\left(\frac{1}{n}\sum_{k=1}^{n}J_n(Y_k)\right) \leq C (n^{2c_1-4\beta
     c_2}+n^{2c_1+2c_2-1}\log{n}).
 \]
 First, we bound the term
 \[
   J_{1n}=\frac{1}{n} \sum_{k=1}^{n} V_1(Y_k)\frac{(f_{Y}(Y_k) -
     \hat{f}_{Y,b}(Y_k))}{\hat{f}_{Y,b}(Y_k)}.
 \] %
 Using the Cauchy-Schwartz inequality, it is straightforward to find
 that
 \begin{align*}
   \V(J_{1n}) %
   & \leq \frac{1}{n^2} \E\left[\left(\sum_{k=1}^{n}
     \frac{g_i(Y_k)g_j(Y_k)}{f^2_{Y,b}} \left(\frac{f_{Y}(y) -
     \hat{f}_{Y,b}(y)}{\hat{f}_{Y,b}(y)}\right) \right)^2
     \right] \\
   & \leq \frac{1}{n}
     \E\left[\sum_{k=1}^{n}\left(\frac{g_i(Y_k)g_j(Y_k)}{f^2_{Y,b}(Y_k)}
     \left(\frac{f_{Y}(Y_k) -
     \hat{f}_{Y,b}(Y_k)}{\hat{f}_{Y,b}(Y_k)}\right)\right)^2\right] & .
 \end{align*}

 By Lemma \ref{lem:f_hat_conv_rate} and Remark \ref{ass:order_conv_h_b}, we
 have %
 \begin{align*}
   \V(J_{1n}) %
   & \leq C \E\left[\frac{1}{n}\sum_{k=1}^{n}
     \frac{g^2_i(Y_k)g^2_j(Y_k)}{f^4_{Y,b}(Y_k)} \right]
     b^{-2}(h^{2\beta} + n^{-1/2}h^{-1}\log{n})^2\\
   &\leq C (n^{2c_1-4\beta c_2}+n^{2c_1-1+2c_2}\log{n}),
 \end{align*}
 where the second inequality is due to the law of large numbers for
 $n^{-1}\sum_{k=1}^{n}R_{i,b}(Y_k)R_{j,b}(Y_k)$.

 For Steps 2 to 4, we denote $Z_k=(\boldsymbol{X}_k,Y_k)$ for
 $k=1,\ldots,n$.

 \paragraph{Step 2.} Prove that
 \[
   \V(V_1) \leq \frac{C}{n}.
 \]


 By the independence of the $Z_k$s and given that $g_j$, $g_l$ and
 $\hat{f}_Y$ are functions built with the second sample, it is clear
 that
 \begin{multline*}
   \V\left(\frac{1}{n} \sum_{k=1}^{n}V_{1}(Y_k)\right) =
   \V\left(\frac{1}{n} \sum_{k=1}^{n} R_{i,b}(Y_k)R_{j,b}(Y_k)\right) \\
   = \frac{1}{n} \V\left(R_{i,b}(Y)R_{j,b}(Y)\right) \leq \frac{C}{n}.
 \end{multline*}

 \paragraph{Step 3.} Show that
 \[
   \V\left(V_{2}\right) + \V\left(V_{3}\right) \leq\frac{C_{1}}{n} +
   \frac{C_{2}}{n^{2}h}.
 \]
 First, we obtain a bound of $\V(V_{2})$.
 Note that
 \begin{align*}
   V_{2}= %
   & \frac{1}{n} \sum_{k=1}^{n} \frac{R_{i,b}(Y_{k})}{f_{Y}(Y_{k})}
     \left( \hat{g}_{j}(Y_{k})-g_{j}(Y_{k}) \right) \\
   & =\frac{1}{n(n-1)} \sum_{k=1}^{n} \sum_{\substack{l=1\\ l\neq k } }^{n}
   \frac{R_{i,b}(Y_{k})}{f_{Y}(Y_{k})}
   \left( X_{jl}K_{h}(Y_{k}-Y_{l}) - g_{j}(Y_{k}) \right) & \\
   & = \frac{1}{n(n-1)} \sum_{k=1}^{n} \sum_{\substack{l=1\\ l\neq k}}^{n}
   \frac{R_{i,b}(Y_{k})}{f_{Y}(Y_{k})} X_{jl}
   K_{h}(Y_{k}-Y_{l}) & \\
   & \hspace{2em} - \frac{1}{n(n-1)} \sum_{i=1}^{n}
     \sum_{\substack{l=1\\ l\neq k }}^{n} R_{i,b}(Y_{k})R_{j,b}(Y_{k}) \\
   & =V_{21}-V_{22}.
 \end{align*}

 Note that
 \[
   \V(V_{22})= \frac{1}{n} \V\left(R_{i,b}(Y)R_{j,b}(Y)\right) =
   \frac{C}{n}.
 \]

 The term $V_{21}$ is indeed a one sample U-statistic of order two.
 Hence, if we define $Z_{li}=(X_{li},Y_{k})$ and rewrite the expression, we
 obtain
 \begin{align*}
   V_{21} %
   & = \frac{1}{n(n-1)} \sum_{k=1}^{n} \sum_{\substack{l=1\\l\neq~k}}^{n}
   \frac{R_{i,b}(Y_{k})}{f_{Y}(Y_{k})}
   X_{jl}K_{h}(Y_{k}-Y_{l}) \\
   & = \frac{1}{2} \binom{n}{2}^{-1}\sum_{(k,l)\in
     C_{2}^{n}}h_{i}(Z_{ik},Z_{il}),
 \end{align*}
 where $C_{2}^{n}=\left\{ (k,l);1\leq k<l\leq n\right\}$ and
 \[
   h_{i}(Z_{jk},Z_{jl}) =
   \frac{R_{i,b}(Y_{k})}{f_{Y}(Y_{k})}X_{jk}K_{h}(Y_{k}-Y_{l}).
 \]
 By employing the symmetric version of $h_{i}$
 \[
   \widetilde{h}_{i}(Z_{jk},Z_{jl}) = \frac{1}{2}\left(
     \frac{R_{i,b}(Y_{k})}{f_{Y}(Y_{k})}X_{jl} +
     \frac{R_{i,b}(Y_{l})}{f_{Y}(Y_{l})} X_{jk}
   \right)K_{h}(Y_{k}-Y_{l}),
 \]
 it is possible (see \cite{kowalski2008modern} or
 \cite{van2000asymptotic}) to decompose $\V(V_{21})$ as
 \begin{multline*}
   \V(V_{21})= \binom{n}{2}^{-1}\left\{ \binom{2}{1} \binom{n-2}{1} \V(
     \E( \widetilde{h}_{i}(Z_{j1},Z_{j2})\vert
     Z_{j1}))\right.\\
   + \left.\binom{2}{2} \binom{n-2}{0} \V(
     \widetilde{h}_{i}(Z_{j1},Z_{j2}))\right\}.
 \end{multline*}

 Since $f_y(y)\leq f_{Y}(y)$, we have
 \begin{align*}
   & \hspace{-2em}
     \E(\widetilde{h}_{i}(Z_{j1},Z_{j2})\vert Z_{j1})\\
   & = \frac{1}{2}\int
     K_{h}(Y_{1}-y)\left(\frac{x_jR_{i,b}(Y_{1})}{f_{Y}(Y_{1})} +
     \frac{X_{j1}R_{i}(y)}{f_{Y}(y)}\right)f(x_j,y)dx_jdy\\
   & = \frac{R_{i,b}(Y_1)}{2f_{Y}(Y_1)} \int K_{h}(Y_{1}-y)
     R_j(y)f_Y(y)dy\\
   & \hspace{2em} + \frac{1}{2}X_{j1}\int K_{h}(Y_{1}-y)
     \frac{R_{i,b}f_Y(y)}{f_{Y}(y)}dy\\
   & \leq \frac{1}{2} R_{i}(Y_1)R_j(Y_1) + \frac{1}{2}X_{j1}R_{i}(Y_1)\\
   & \hspace{2em} + \frac{R_{i}(Y_1)}{2f_{Y}(Y_1)} \int
     K_{h}(Y_{1}-y) \{R_j(y)f_Y(y) - R_j(Y_1)f_Y(Y_1)\} dy\\
   & \hspace{2em} + \frac{1}{2}X_{j1} \int K_{h}(Y_{1}-y) \left\{
     R_{i}(y) - R_{i}(Y_1)\right\} dy\\
   & \leq \frac{1}{2}R_{i}(Y_1)R_j(Y_1) + \frac{1}{2}X_{j1}R_{i}(Y_1) +
     J_{1}(Z_{j1}) + J_{2}(Z_{j1}).
 \end{align*}

 Using Assumption \ref{ass:f_g_Holder} and by applying the same arguments
 employed in the proof of Lemma \ref{lem:g_hat_conv_rate}, we can
 conclude that
 \begin{align*}\label{pb:prob_hyp_g}
   \V(J_{1}(Z_{j1})) %
   & \leq\E[(J_{1}(Z_{j1}))^{2}] \\
   & \leq C h^{2\beta}\int\left(\frac{R_{i}(y)}{f_{Y}(y)}\right)^2
     \left(\int\vert u^s K(u)du\vert\right)^{2}f_{Y}(y)dy\nonumber \\
   &\leq C h^{2\beta}\E[R_i^2(Y_1)]\leq Ch^{2\beta}.
 \end{align*}
 Moreover, as $f\in \mathcal{H}(\beta,L)$ and $0<\eta<f_Y(y)$, we have
 \begin{align*}
   \V(J_2(Z_{j1})) & \leq\E[(J_{2}(Z_{j1}))^{2}] \\
                 & \leq \frac{1}{4} \E\left[X_{j1}^2 \left(\int
                   K(u)\left(R_{i}(Y_{1}+uh)-R_{i}(Y_{1})du\right) du
                   \right)^2\right]\\
                 & \leq Ch^{2\beta}
 \end{align*}


 and thus,
 \begin{align*}
   \V(\E(\widetilde{h}_{i}(Z_{j1},Z_{j2})\vert Z_{j1})) &
                                                      \leq\V\left(\frac{1}{2}R^{2}(Y_{1}) + \frac{1}{2}X_{l1}R(Y_{1})\right) + C_1
                                                      h^{2\beta}. + \frac{C_2 }{n}
 \end{align*}
 By similar calculations, we bound
 \[
   \V(\widetilde{h}_{i}(Z_{j1},Z_{j2}))\leq\frac{C_{2}}{h}.
 \]
 Using the same procedure, we can bound $\V(V_3)$.
 We conclude that


\begin{align*}
  \V(V_{2}) + \V(V_{3}) & \leq\frac{2}{n(n-1)}\left\{ (n-2)\left(C_{1}
                        + C_2 h^{2\beta} + \right) + \frac{C_{4}}{h}\right\} \\
                      & \leq C_1 h^{2\beta} + \frac{C_{2}}{n^{2}h}.
\end{align*}

\paragraph{Step 4.} Show that
\[
  \V(V_4) \leq C (h^{4\beta} + n^{-2}h^{-4}\log^4{n}).
\]

 %
 %
 %
 %
Using Lemma \ref{lem:bound_g_hat_minus_g}, we obtain
\begin{align*}
  \V\left(V_{4}\right) %
  & \leq\E[V_{4}^{2}] \\
  & =\frac{1}{n^{2}} \E\left[\left(\sum_{k=1}^{n} \frac{
    \left(\hat{g}_{i}(Y_{k})-g_{i}(Y_{k})\right)
    \left(\hat{g}_{j}(Y_{k})-g_{j}(Y_{k})\right)}{f_{Y}^{2}(Y_{k})}
    \right)^{2}\right]\\
  & \leq \frac{Cn^2}{n^{2}} (h^{\beta} + n^{-1/2}h^{-1}\log{n})^4\\
  & \leq C (h^{\beta} + n^{-1/2}h^{-1}\log{n})^4\\
  & \leq C (h^{4\beta} + n^{-2}h^{-4}\log^4{n}).
\end{align*}

\paragraph{Final bound}

By combining all of the previous results, we have

\begin{equation*}
  \V\left(\hat{\sigma}_{ij}\right)\leq
  C_1 h^{2\beta} + \frac{C_{2}\log^{4}n}{n^{2}h^{4}} + \frac{1}{n}.
\end{equation*}

\end{proof}

\begin{proof}[Proof of Lemma~\ref{lem:bound_g_hat_minus_g}]
  The kernel function $K$ is uniformly continuous on $[-1,1]$, so by
  writing $c_{1}=\sup_{\vert u\vert\leq 1}\vert K(u)\vert$, we have
  \[
    \sup_{y}
    \left(\V\left(K\left(\frac{y-Y}{h}\right)\right)\right)^{1/2}
    \leq\sup_{y}\left(\int
      K^{2}\left(\frac{y-Y}{h}\right)f(Y)dY\right)^{1/2} \leq c_{1}.
  \]
  Choose $\varepsilon=\log n$, then as $n\to\infty$, we have
  \[
    \sup_{y}\vert\hat{g}_{i}(y) - \E\left(\hat{g}_{i}(y)\right)\vert =
    O_{p}\left(n^{-1/2}h^{-1}\log n\right).
  \]
  By contrast, we expand $g_{i}(y)$ in a Taylor series with the Lagrange
  form of the remainder term (see \cite{prakasa1983nonparametric}, page
  47).
  Using Assumption \ref{ass:f_g_Holder} and Remark \ref{rem:g_Holder}, for
  any $0 < \tau < 1$ and $s=\left\lfloor\beta\right\rfloor$, we have,

 \begin{align*}
   & \sup_{y}\vert\E\left(\hat{g}_{i}(y)\right) - g_{i}(y)\vert \\
   & =\sup_{y}\left|\int K_{h}(y-Y)\left\{ g_{i}(Y)-g_{i}(y)\right\}
     dY\right|\\
   & =\sup_{y}\left|\int K(u)\left\{ g_{i}(y + uh)-g_{i}(y)\right\}
     du\right|\\
   & =\sup_{y}\left|\int K(u)\left\{ g_{i}(y) + uhg_{i}^{\prime}(y) + \cdots
     + (uh)^{s}\frac{g_{i}^{(s)}(y + \tau uh)}{s!}-g_{i}(y)\right\}
     du\right|\\
   & =\sup_{y}\left|\int
     K(u)\frac{\left(uh\right)^{s}}{s!}\left(g_{i}^{(s)}(y + \tau
     uh)-g_{i}^{(s)}(y)\right)du\right|,
 \end{align*}

 and as $g_i \in \mathcal{H}(\beta,L)$, we conclude that,
 \[
   \sup_{y}\vert\E\left(\hat{g}_{i}(y)\right)-g_{i}(y)\vert \leq
   c\int\left|u^{\beta}K(u)\right|du\cdot h^{\beta}.\qedhere
 \]
\end{proof}

\section*{References}

\bibliographystyle{apalike}
\bibliography{biblio}

\end{document}